\documentclass[11pt]{article}
\usepackage[T1]{fontenc}
\usepackage{lmodern}
\usepackage{fullpage}
\usepackage{amsmath,amsthm,amssymb,braket,color,mathtools}
\usepackage{thmtools,thm-restate}
\usepackage{hyperref}
\hypersetup{colorlinks=true, citecolor=blue, linkcolor=red, urlcolor=blue}
\usepackage[capitalise]{cleveref}

\usepackage{complexity}

\usepackage{threeparttable}

\declaretheorem{theorem}
\declaretheorem[sibling=theorem]{lemma}

\declaretheorem[sibling=theorem]{fact}
\declaretheorem[sibling=theorem]{definition}

\newcommand{\littleo}[1]{o\left(#1\right)}
\newcommand{\bigo}[1]{\mathcal{O}\left(#1\right)}
\newcommand{\tildeo}[1]
{\tilde{\mathcal{O}}\left(#1\right)}
\newcommand{\bigtheta}[1]{\Theta\left(#1\right)}
\newcommand{\bigomega}[1]{\Omega\left(#1\right)}

\title{Nearly optimal quantum circuits for Boolean oracles}
\author{
    Junhong Nie
    \thanks{
        Shandong University.
        Email: \nolinkurl{123miyoi@gmail.com}
    }
    \and
    Wei Zi
    \thanks{
        Quantum  Science  Center  of Guangdong-Hong Kong-Macao Greater Bay Area.
        Email: \nolinkurl{ziwei.quantum@outlook.com}
    }
}
\date{}

\begin{document}

\maketitle

\begin{abstract}
    Quantum oracle of Boolean functions is one of the central bridges between classical and quantum algorithms, but the study focusing at quantum circuit optimization of such oracle is yet closed. In this paper, we propose nearly optimal tradeoffs among circuit size, circuit depth and ancilla count, for quantum oracles of three kinds of Boolean functions:
    \begin{itemize}
        \item general total Boolean functions with output size $b$: with $1\le m\le\bigtheta{\frac{2^n}{n}}$ ancilla, size $\bigo{\frac{b2^n}{\log(n+m)}}$, depth $\bigo{\frac{b2^n}{n+m}}$;
        \item partial Boolean functions of effective support size $d$ and output size $b$: with $\bigtheta{\log d}\le m\le \bigtheta{d}$ ancilla, size $\bigo{n\log d+bd}$, depth $\bigo{\frac{n\log n\log d}{n+m}+\log n+\frac{d(\log d+b\log m)}{m}}$;
        \item sparse Boolean functions of true input size $d$: with $\bigtheta{\log n+\log d}\le m\le\bigtheta{\frac{nd}{\log d}}$ ancilla, size $\bigo{n^2\log d+\frac{nd}{\log(\log d+m/n)}}$, depth $\bigo{\frac{n^2\log n\log d}{n+m}+\log n+\frac{nd}{m}}$.
    \end{itemize}
    All the size and depth bounds are asymptotically optimal up to logarithmic factors in the corresponding ancilla count regions. We hope these results find applications in scenarios where classical procedures are needed to be embedded into quantum circuits, such as QROM implementation and quantum algorithm design.
\end{abstract}

\section{Introduction}\label{sec:intro}
Given the truth table of a Boolean function $f:\{0,1\}^n\to\{0,1\}$, the quantum oracle of $f$ computes
\[\sum_{x}\ket{x}\bra{x}\otimes X^{f(x)}.\]
The importance of Boolean function oracle in quantum computing is self-evident. It provides basic interface between classical data and quantum algorithms, e.g. quantum query model \cite{grover1996Fast}, look-up table \cite{zhu2025Unified}, quantum read-only memory (QROM) \cite{babbush2018Encoding}. On the other hand, such oracle serves as central sub-procedure in many quantum algorithms, e.g. Hamiltonian simulation \cite{low2017Optimal,low2019Hamiltonian}, cryptanalysis \cite{grover1996Fast,jaques2020Implementing}, data-intensive scenarios \cite{an2022Quantum}. In both cases, when oracles are not treated as black boxes, the synthesis cost can sometimes dominate the overall resource. Thus, understanding the tradeoff between circuit size, depth and ancilla count is essential for quantum algorithm design and analysis.

In this paper, we consider implementation under the quantum circuit model that consists of constant size gate set with constant gate width, such as the normal Clifford+$T$ gate set in fault-tolerant quantum computation. We provide nearly optimal circuit size, depth and ancilla count tradeoffs for quantum oracle of total, partial and sparse Boolean functions in large ancilla count regions: All constructions have size asymptotically optimal up to logarithmic factors, and the product of depth and ancilla count is also the same order as size up to logarithmic factors.

\subsection{Main results}
For arbitrary total Boolean function of input size $n$, we propose a complete tradeoff scheme for the quantum oracle with arbitrary ancilla count in \cref{thm:general_tradeoff}.
\begin{restatable}{theorem}{thmgeneraltradeoff}\label{thm:general_tradeoff}
    Any Boolean function $f:\{0,1\}^n\to\{0,1\}^b$ can be implemented in $\bigo{\frac{b2^n}{\log(n+m)}}$ size, $\bigo{\frac{b2^n}{n+m}}$ depth, with $1\le m\le\bigtheta{\frac{2^n}{n}}$ ancilla.
\end{restatable}

With the help of \cref{thm:general_tradeoff}, we can obtain nearly optimal tradeoffs for quantum oracle of reversible Boolean functions:
\begin{restatable}{corollary}{correversibletradeoff}\label{cor:reversible_tradeoff}
    Any reversible function $f:\{0,1\}^n\to\{0,1\}^n$ can be implemented in $\bigo{\frac{n2^n}{\log(n+m)}}$ size, $\bigo{\frac{n2^n}{n+m}}$ depth, with $n+1\le m\le\bigtheta{\frac{2^n}{n}}$ ancilla.
\end{restatable}

Partial Boolean functions are those whose size of effective support is limited. The function restricts value only on a subset of all $n$-bit inputs, and can be arbitrary on the rest. We propose the depth-ancilla tradeoff for partial Boolean oracle in \cref{thm:partial_tradeoff}.
\begin{restatable}{theorem}{thmpartialtradeoff}\label{thm:partial_tradeoff}
    Suppose the partial Boolean function $f:\{0,1\}^n\to\{0,1,*\}^b$ with effective input $X=\{x_1,\dots,x_d\}\subseteq\{0,1\}^n$. Then the oracle of $f$ can be implemented in size $\bigo{n\log d+bd}$, depth $\bigo{\frac{n\log n\log d}{n+m}+\log n+\frac{d(\log d+b\log m)}{m}}$, with $\bigtheta{\log d}\le m\le \bigtheta{d}$ ancilla.
\end{restatable}

Sparse Boolean functions are those whose size of true input is limited. That is, the function takes value $1$ only on a subset of all $n$-bit inputs, and takes $0$ on the rest. We propose the depth-ancilla tradeoff for sparse Boolean oracle in \cref{thm:sparse_tradeoff}. The size can be further improved, which we refer to \cref{thm:sparse_tradeoff_improved}.
\begin{restatable}{theorem}{thmsparsetradeoff}\label{thm:sparse_tradeoff}
    Suppose the sparse total Boolean function $f:\{0,1\}^n\to\{0,1\}$ with $f^{-1}(1)=X=\{x_1,\dots,x_d\}$. Then the oracle of $f$ can be implemented in size $\bigo{n^2\log d+\frac{nd}{\log(\log d+m/n)}}$, depth $\bigo{\frac{n^2\log n\log d}{n+m}+\log n+\frac{nd}{m}}$, with $\bigtheta{\log n+\log d}\le m\le\bigtheta{\frac{nd}{\log d}}$ ancilla.
\end{restatable}

\subsection{Related works}
\paragraph*{General total Boolean functions}
It is well-known that any Boolean function $f:\{0,1\}^n\to\{0,1\}$ has Boolean circuit of size $(1+\littleo{1})\frac{2^n}{n}$ \cite{lupanov1970Method} which is optimal according to Shannon's counting argument \cite{shannon1949synthesis}, and such construction can be naturally translated into quantum circuit of the same asymptotic order. Besides, several worst-case constructions of general Boolean oracle are present in folklore (which we exhibit in \cref{sec:general}), but none of them are optimal in circuit size. There are also plenty of works aiming to heuristically reduce circuit cost of Boolean oracles, e.g. \cite{wille2009BDDbased,meuli2019Role,soeken2019LUTBased,phalak2022Optimization}, but lack of worst-case guarantee.  \cref{thm:general_tradeoff} essentially closes the gap of circuit size, depth and ancilla tradeoff of general Boolean oracle. The comparison between \cref{thm:general_tradeoff} and known worst-case constructions is shown in \cref{tab:general_literature}.
\begin{table}
    \centering
    \begin{threeparttable}
        \begin{tabular}{l|c|c|c}
            \hline
            &size &depth &\#ancilla\\
            \hline lower bound, \cref{thm:lower_bound_general} &$\bigomega{\frac{2^n}{\log\min\{n+m,2^n\}}}$ &$\bigomega{\max\{\frac{2^n}{(n+m)\log(n+m)},n\}}$ &$m$\\
            enumerate monomials &$\bigo{n2^n}$ &$\bigo{n2^n}$ &$1$\\
            Fourier expansion\tnote{$\ast$} &$\bigo{2^n}$ &$\bigo{2^n}$ &$1$\\
            \cref{thm:general_naive_limited} &$\bigo{2^n}$ &$\bigo{2^n}$ &$n$\\
            \cref{thm:general_optimal_size} \cite{lupanov1970Method} &$\bigo{\frac{2^n}{n}}$ &$\bigo{n}$ &$\bigo{\frac{2^n}{n}}$\\
            \cref{thm:general_naive_rich} &$\bigo{2^n}$ &$\bigo{n}$ &$2^n$\\
            our result, \cref{thm:general_tradeoff} &$\bigo{\frac{2^n}{\log \min\{n+m,2^n\}}}$ &$\bigo{\max\{\frac{2^n}{n+m},n\}}$ &$1\le m\le\bigtheta{\frac{2^n}{n}}$\\
            \hline
        \end{tabular}
        \begin{tablenotes}
            \item[$\ast$] Not full reversible logic.
        \end{tablenotes}
        \caption{Known constructions for general Boolean oracle.}
        \label{tab:general_literature}
    \end{threeparttable}
\end{table}

\paragraph*{Reversible Boolean functions}
A natural special case of Boolean oracle is the quantum oracle of reversible Boolean functions $\{0,1\}^n\to\{0,1\}^n$ which computes the output onto the input wires. With $m$ ancilla, Zakablukov proves a lower bound of size $\bigomega{\frac{n2^n}{\log(n+m)}}$ and depth $\bigomega{\frac{n2^n}{(n+m)\log(n+m)}}$ under the restriction of reversible circuit implementations \cite{zakablukov2017asymptotic}. Zakablukov also gives elegant constructions in different regions of ancilla count, but only partially matching the lower bound. \cref{cor:reversible_tradeoff} essentially closes the circuit size and depth gap in a broad ancilla count region $n<m\le\bigtheta{\frac{2^n}{n}}$. This is summarized in \cref{tab:reversible_literature}.
\begin{table}
    \centering
    \begin{tabular}{l|c|c|c}
        \hline
        &size &depth &\#ancilla\\
        \hline reversible lower bound \cite{zakablukov2017asymptotic} &$\bigomega{\frac{n2^n}{\log(n+m)}}$ &$\bigomega{\frac{n2^n}{(n+m)\log(n+m)}}$ &$m$\\
        zero ancilla \cite{zakablukov2017asymptotic,wu2024Asymptotically} &$\bigo{\frac{n2^n}{\log n}}$ &$\bigo{\frac{n2^n}{\log n}}$ &$0$\\
        moderate rich ancilla \cite{zakablukov2016Asymptotic} &$\bigo{2^n}$ &$\bigo{n}$ &$\bigo{2^n}$\\
        tradeoff \cite{zakablukov2017General} &$\bigo{2^n+\frac{n2^n}{\log(m/n)}}$ &$\bigo{2^n\log\frac{n}{\log(m/n)}}$ &$8n<m\le n2^{n-o(n)}$\\
        our result, \cref{cor:reversible_tradeoff} &$\bigo{\frac{n2^n}{\log(n+m)}}$ &$\bigo{\frac{n2^n}{n+m}}$ &$n<m\le\bigtheta{\frac{2^n}{n}}$\\
        \hline
    \end{tabular}
    \caption{Known constructions for reversible Boolean oracle.}
    \label{tab:reversible_literature}
\end{table}

\paragraph*{Sparse Boolean functions}
\cref{thm:lower_bound_sparse} shows that almost all sparse Boolean functions require size $\bigomega{\frac{d(n-\log d)}{\log\min\{W,d(n-\log d)\}}}$ and depth $\bigomega{\max\left\{\frac{d(n-\log d)}{W\log W},\log d+\log(n-\log d)\right\}}$, where $W=n+1+m$. At the high-ancilla end, \cite{zhang2022Quantum} constructs any sparse Boolean oracle in size $\bigo{nd}$ and depth $\bigo{\log(nd)}$ using $\bigo{nd}$ ancilla. \cite{zhang2024Circuit} then gives a depth--ancilla tradeoff of size $\bigo{nd}$ and depth $\bigo{\frac{nd\log m}{m}}$ for $\bigtheta{n}\le m\le\bigtheta{nd}$. \Cref{thm:sparse_tradeoff} pushes this tradeoff down to $\bigtheta{\log n+\log d}\le m\le\bigtheta{\frac{nd}{\log d}}$ ancilla and approaches the lower bounds up to logarithmic factors when $d=\bigomega{n}$. Finally, with $\bigtheta{n}$ additional ancilla, \cref{thm:sparse_tradeoff_improved} gives a further logarithmic improvement to the leading size and depth terms.

\paragraph*{Quantum read-only memory}
Quantum read-only memory (QROM) essentially computes Boolean functions $\{0,1\}^n\to\{0,1\}^b$ \cite{babbush2018Encoding}, which is widely used in quantum algorithm frameworks like Hamiltonian simulation and linear combination of unitary. There have been works optimizing the quantum circuit implementation of QROM, e.g. \cite{babbush2018Encoding,meuli2019Role,low2024Trading,motlagh2026Halving}. Our results can be viewed as circuit size and depth efficient implementations of different type of QROMs.

\paragraph*{Fault-tolerant cost optimization}
Optimizing $T$ cost under the Clifford+$T$ model is essential for fault-tolerant quantum computation, since it is notoriously harder to employ $T$ gates than Clifford gates in current fault-tolerant schemes. There are many representative works reducing the $T$ count \cite{babbush2018Encoding,boyar2018Multiplicative,meuli2019Role,low2024Trading,motlagh2026Halving} or $T$ depth \cite{low2024Trading,dutta2025Optimal} of Boolean oracle or QROM or SELECT operator. Our results can also be interpreted as optimization of $T$ depth or Toffoli depth. Also, we emphasize that optimizing the total circuit size and depth is still of great significance, especially when the $T$ cost is asymptotically much smaller than the Clifford cost.

\section{Preliminaries}\label{sec:pre}
\subsection{Definitions}
We first clarify the objects studied in this paper. A quantum oracle of Boolean function computes the function into a target qubit.
\begin{definition}[Boolean oracle]
    A quantum oracle of total Boolean function $f:\{0,1\}^n\to\{0,1\}^b$ computes
    \[\ket{x}\ket{t}\to\ket{x}\ket{t\oplus f(x)}\]
    for all $x\in\{0,1\}^n$.
\end{definition}
A partial Boolean oracle relaxes the domain on which the oracle must be correct, i.e. it computes a partial Boolean function.
\begin{definition}[Partial Boolean oracle]
    Let $X=\{x_1,\dots,x_d\}\subseteq\{0,1\}^n$, and partial Boolean function $f:\{0,1\}^n\to\{0,1,*\}^b$ such that the effective input is $X$, that is, $f(X)\subseteq\{0,1\}^b$ and $f(\{0,1\}^n\setminus X)=\{*^b\}$. A partial quantum oracle of $f$ computes
    \[\ket{x_i}\ket{t}\to\ket{x_i}\ket{t\oplus f(x_i)}\]
    for all $x_i\in X$.
\end{definition}
A sparse Boolean oracle computes a ``sparse'' total Boolean function $f$ which has $\#f^{-1}(1)=d$.
\begin{definition}[Sparse Boolean oracle]
    Let $X=\{x_1,\dots,x_d\}\subseteq\{0,1\}^n$, and total Boolean function $f:\{0,1\}^n\to\{0,1\}$ such that $f^{-1}(1)=X$. A sparse quantum oracle of $f$ computes
    \[\ket{x}\ket{b}\to\ket{x}\ket{b\oplus f(x)}\]
    for all $x\in\{0,1\}^n$.
\end{definition}

All our constructions uses clean ancilla: the ancilla register is initialized at state $\ket{0}$ and returned to $\ket{0}$ after computation, so the circuit implements the intended unitary on the input register exactly.
\begin{definition}[Implementation with $m$ ancilla]
    A unitary matrix $U$ on $n$ qubits is implemented with $m$ (clean) ancilla if there is a quantum circuit $C$ on $n+m$ qubits such that
    \[C(\ket{\psi}\otimes\ket{0^m})=(U\ket{\psi})\otimes\ket{0^m}\]
    for every $n$-qubit state $\ket{\psi}$.
\end{definition}

\subsection{Low-depth quantum circuit primitives}
Below we record some fundamental low-depth primitives.
\begin{fact}\label{fact:fan}
    Quantum fan-out gates computing
    \[\ket{c}\ket{x_1,\dots,x_n}\to\ket{c}\ket{x_1\oplus c,\dots,x_n\oplus c}\]
    and quantum fan-in gates computing
    \[\ket{x_1,\dots,x_n}\ket{t}\to\ket{x_1,\dots,x_n}\ket{t\oplus\bigoplus_{i=1}^nx_i}\]
    can both be implemented in $\bigo{n}$ size, $\bigo{\log n}$ depth with $0$ ancilla.
\end{fact}
\begin{fact}\label{fact:Toffoli_fanin}
    Toffoli fan-in gates computing
    \[\ket{x_1,\dots,x_n}\ket{y_1,\dots,y_n}\ket{t}\to\ket{x_1,\dots,x_n}\ket{y_1,\dots,y_n}\ket{t\oplus\bigoplus_{i=1}^n(x_i\wedge y_i)}\]
    can be implemented in $\bigo{n}$ size, $\bigo{\log n}$ depth with $0$ ancilla.
\end{fact}
\begin{theorem}[\cite{nie2024Quantum}]\label{thm:Toffoli}
    An $n$-controlled Toffoli gate can be implemented in $\bigo{n}$ size, $\bigo{\log n}$ depth with $1$ ancilla.
\end{theorem}

\section{Lower bounds via simple counting argument}
In this paper we fix a finite gate set with finite gate width such as Clifford+$T$. Since all these Boolean oracles can be implemented with no error, we state the general lower bound without approximation.

\subsection{General family of Boolean functions}
Below we fix a finite gate set $\mathcal{G}$ whose gates act on at most $k\ge 2$ qubits. Given a family $\mathcal{F}$ of Boolean functions $\{0,1\}^n\to\{0,1\}^b$. Let $H(\mathcal{F})=\log\#\mathcal{F}$. Let $S_\mathcal{G}(f,m)$ and $D_\mathcal{G}(f,m)$ be the optimal size and depth of $f\in\mathcal{F}$ with $m$ ancilla. 
\begin{lemma}\label{lem:lower_bound_size}
    For every $s>0$, define $W_s=n+b+\min\{m,ks\}$, then
    \[\#\{f\in\mathcal{F}:S_\mathcal{G}(f,m)\le s\}\le (s+1)\left(\#\mathcal{G}\cdot W_s^k\right)^s.\]
    Consequently, for every $0<\epsilon<1$, let $W_H=n+b+\min\{m,kH(\mathcal{F})\}$, then
    \[\Pr_{f\sim\mathcal{F}}\left[S_\mathcal{G}(f,m)\le\frac{(1-\epsilon)H(\mathcal{F})}{1+\log\#\mathcal{G}+k\log W_H}\right]\le 2^{-\epsilon H(\mathcal{F})}.\]
\end{lemma}
\begin{proof}
    The number of circuits of size $t$ is at most $\left(\#\mathcal{G}\cdot W_s^k\right)^t$, thus the number of circuits of size $\le s$ is at most
    \[\sum_{t=0}^s\left(\#\mathcal{G}\cdot W_s^k\right)^t\le(s+1)\left(\#\mathcal{G}\cdot W_s^k\right)^s,\]
    this proves the first claim.
    
    Let $s_\epsilon=\frac{(1-\epsilon)H(\mathcal{F})}{1+\log\#\mathcal{G}+k\log W_H}$, notice that $\log(s_\epsilon+1)\le s_\epsilon$ and $W_{s_\epsilon}\le W_H$, so
    \[\Pr_{f\sim\mathcal{F}}\left[S_\mathcal{G}(f,m)\le s_\epsilon\right]\le 2^{\log(s_\epsilon+1)+s_\epsilon(\log\#\mathcal{G}+k\log W_{s_\epsilon})-H(\mathcal{F})}\le 2^{-\epsilon H(\mathcal{F})}.\]
\end{proof}
\begin{lemma}\label{lem:lower_bound_depth_1}
    Let $W=n+b+m$. Then for every $d>0$,
    \[\#\{f\in\mathcal{F}:D_\mathcal{G}(f,m)\le d\}\le(Wd+1)\left(\#\mathcal{G}W^k\right)^{Wd}.\]
    Consequently, for every $0<\epsilon<1$,
    \[\Pr_{f\sim\mathcal{F}}\left[D_\mathcal{G}(f,m)\le\frac{(1-\epsilon)H(\mathcal{F})}{W(1+\log\#\mathcal{G}+k\log W)}\right]\le 2^{-\epsilon H(\mathcal{F})}.\]
\end{lemma}
\begin{proof}
    The proof is similar to \cref{lem:lower_bound_size}, by noticing that the size of a quantum circuit with depth $d$ and width $W$ is no larger than $Wd$.
\end{proof}
When the ancilla count is extremely large, \cref{lem:lower_bound_depth_1} becomes trivial, where the light cone argument gives basic guarantee of depth lower bound.
\begin{lemma}\label{lem:lower_bound_depth_2}
    For every $d>0$, define $W_d=n+b+\min\{m,bk^d\}$ and $S_d=b\frac{k^d-1}{k-1}$, then
    \[\#\{f\in\mathcal{F}:D_\mathcal{G}(f,m)\le d\}\le(S_d+1)\left(\#\mathcal{G}W_d^k\right)^{S_d}.\]
    Consequently, whenever $(S_d+1)\left(\#\mathcal{G}W_d^k\right)^{S_d}\le(\#\mathcal{F})^{(1-\epsilon)}$,
    \[\Pr_{f\sim\mathcal{F}}\left[D_\mathcal{G}(f,m)\le d\right]\le 2^{-\epsilon H(\mathcal{F})}.\]
    Specifically, when $n+b+m\ge bk^d$ and $d\le\bigtheta{\log\frac{H(\mathcal{F})}{b\log H(\mathcal{F})}}$,
    \[\Pr_{f\sim\mathcal{F}}\left[D_\mathcal{G}(f,m)\le d\right]\le 2^{-\bigomega{H(\mathcal{F})}}.\]
\end{lemma}
\begin{proof}
    From the $b$ output qubits, the causal sub-circuit involves at most $\min\{bk^d,n+b+m\}\le W_d$ qubits, and the size of this causal sub-circuit is at most $b\sum_{j=0}^{d-1}k^j=S_d$. Thus, a similar counting argument like \cref{lem:lower_bound_size} proves the first claim. The rest claims are straightforward.
\end{proof}

\subsection{Lower bounds of quantum oracles}
Now, we provide lower bounds for the circuit size and depth of general total, partial and sparse Boolean oracles. All these lower bounds aim to quantum circuits using constant size gate set with constant width, such as Clifford+$T$. These results are merely restatements of \cref{lem:lower_bound_size,lem:lower_bound_depth_1,lem:lower_bound_depth_2}, so we omit the proofs.
\begin{theorem}\label{thm:lower_bound_general}
    For all but $2^{-\bigomega{b2^n}}$ fraction of total Boolean functions $\{0,1\}^n\to\{0,1\}^b$, let $W=n+b+m$, the quantum oracle computing that function requires size $\bigomega{\frac{b2^n}{\log\min\{W,b2^n\}}}$ and depth $\bigomega{\max\left\{\frac{b2^n}{W\log W},n-\log(n+\log b)\right\}}$, with $m$ ancilla.
\end{theorem}

\begin{theorem}\label{thm:lower_bound_partial}
    Suppose $n\le(bd)^{\bigo{1}}$. Fix an effective support $X\subseteq\{0,1\}^n$ of size $d$. For all but $2^{-\bigomega{bd}}$ fraction of partial Boolean functions on $X$ with output size $b$, let $W=n+b+m$, the quantum oracle computing that function requires size $\bigomega{\frac{bd}{\log\min\{W,bd\}}}$ and depth $\bigomega{\max\left\{\frac{bd}{W\log W},\log d-\log\log(bd)\right\}}$, with $m$ ancilla.
\end{theorem}

\begin{theorem}\label{thm:lower_bound_sparse}
    Suppose $d\le 2^{n-1}$. For all but $2^{-\bigomega{d(n-\log d)}}$ fraction of sparse Boolean functions of true input size $d$, let $W=n+1+m$, the quantum oracle computing that function requires size $\bigomega{\frac{d(n-\log d)}{\log\min\{W,d(n-\log d)\}}}$ and depth $\bigomega{\max\left\{\frac{d(n-\log d)}{W\log W},\log d+\log(n-\log d)\right\}}$, with $m$ ancilla.
\end{theorem}

\section{General Boolean oracle}\label{sec:general}
In this section we propose constructions of general Boolean oracle. We first prove in \cref{ssec:general_optimal_size} that any Boolean function $f:\{0,1\}^n\to\{0,1\}$ can be implemented in $\bigo{\frac{2^n}{n}}$ size, $\bigo{n}$ depth with $\bigo{\frac{2^n}{n}}$ ancilla. In \cref{ssec:general_tradeoff}, we prove \cref{thm:general_tradeoff}.

\subsection{Naive constructions for general Boolean oracle}
A naive method to implement $f$ is to view $f$ as a polynomial over $\mathbb{F}_2$, and compute every monomial and sum them up. Denote $\pi_S=\bigwedge_{i\in S}x_i$, then
\begin{equation}\label{eq:general_polynomial}
    f=\bigoplus_{S\subseteq[n]}a_S\wedge\pi_S,
\end{equation}
where $a_S\in\{0,1\}$ is the coefficient of $\pi_S$.

With rich ancilla, we can compute all $\pi_S$ in small depth and then compute $f$ according to \cref{eq:general_polynomial}.
\begin{lemma}\label{lem:general_monomial}
    There is a quantum circuit computing all $\{\pi_S:S\subseteq[n]\}$ in $\bigo{2^n}$ size, $\bigo{n}$ depth with $2^n$ ancilla.
\end{lemma}
\begin{proof}
    First, copy $x_i$ for $2^{i-1}$ times using quantum fan-out gates in \cref{fact:fan}, which has total size $\bigo{2^n}$, depth $\bigo{n}$, with $2^n$ ancilla. After that, compute $\pi_S$ recursively. Suppose all $\{\pi_S:S\in[k-1]\}$ are computed, then to compute all $\{\pi_S:S\in[k]\}$, $2^{k-1}$ Toffoli gates in depth $1$ suffices. In all, the second step has size $\bigo{2^n}$, depth $\bigo{n}$ with $0$ ancilla.
\end{proof}

\begin{lemma}\label{lem:general_from_monomial}
    Given monomials $\pi_S$ for all $S\subseteq[n]$, any $f:\{0,1\}^n\to\{0,1\}$ can be computed in $\bigo{2^n}$ size and $\bigo{n}$ depth, with $0$ ancilla.
\end{lemma}
\begin{proof}
    According to \cref{eq:general_polynomial}, it suffices to sum up all $\pi_S$ such that $a_S=1$ using fan-in gates in \cref{fact:fan}.
\end{proof}

\begin{theorem}\label{thm:general_naive_rich}
    Any Boolean function $f:\{0,1\}^n\to\{0,1\}$ can be implemented in $\bigo{2^n}$ size and $\bigo{n}$ depth, with $2^n$ ancilla.
\end{theorem}
\begin{proof}
    Combine \cref{lem:general_monomial,lem:general_from_monomial}.
\end{proof}

When the ancilla number is limited to $\bigo{n}$, it is still possible to implement Boolean oracle in $\bigo{2^n}$ size by using \cref{eq:general_polynomial}.
\begin{theorem}\label{thm:general_naive_limited}
    Any Boolean function $f:\{0,1\}^n\to\{0,1\}$ can be implemented in $\bigo{2^n}$ size and depth, with $n$ ancilla.
\end{theorem}
\begin{proof}
    According to \cref{eq:general_polynomial}, $f$ can be written as
    \[x_1f_1(x_2,\dots,x_n)\oplus f_2(x_2,\dots,x_n).\]
    We prove recursively the slightly stronger statement that, for a control qubit $c$, we can implement
    \[
        \ket{c}\ket{x_1,x_{[2\dots n]}}\ket{0^n}\ket{t}
        \mapsto
        \ket{c}\ket{x_1,x_{[2\dots n]}}\ket{0^n}
        \ket{t\oplus cf(x_1,\dots,x_n)}.
    \]
    We also allow $c=1$ to be hard-wired. The case $n=0$ is immediate.

    Let $a_1$ be the first ancillary qubit. Compute $a_1\mapsto a_1\oplus cx_1$, recursively add $a_1f_1(x_2,\dots,x_n)$ to $t$ using the other $n-1$ ancillary qubits, and then uncompute $a_1$. Next, recursively add $cf_2(x_2,\dots,x_n)$ to $t$ using the same $n-1$ ancillary qubits. Thus, the target is changed by
    \[
        cx_1f_1(x_2,\dots,x_n)\oplus cf_2(x_2,\dots,x_n)
        =cf(x_1,\dots,x_n),
    \]
    and all ancillary qubits are returned to $\ket{0}$. At each level there are two recursive calls on $n-1$ variables and $\bigo{1}$ additional gates. Therefore, the construction has size and depth $\bigo{2^n}$, with $n$ ancilla. Taking the hard-wired control $c=1$ proves the theorem.
\end{proof}

\subsection{Optimal size with rich ancilla}\label{ssec:general_optimal_size}
To push the size to optimal, we need the observation that by regrouping the formula in \cref{eq:general_polynomial}, there are many identical small Boolean functions. To be concrete, let $p,q\in[n]$ and $p+q=n$. It is clear that
\begin{equation}\label{eq:general_partial}
    f=\bigoplus_{S\subseteq[p]}(\pi_S\wedge f_S(x_{p+1},\dots,x_n)),
\end{equation}
where $\pi_S=\bigwedge_{i\in S}x_i$ and $f_S$ is a Boolean function on $q$ variables. Notice that there are only $2^{2^q}$ distinct Boolean functions of input size $q$, so instead of computing $f_S$ for each $S$, we compute all Boolean functions on $(x_{p+1},\dots,x_n)$, which results in a lower size. This construction is due to Lupanov \cite{lupanov1970Method}.

\begin{theorem}[\cite{lupanov1970Method}]\label{thm:general_optimal_size}
    Any Boolean function $f:\{0,1\}^n\to\{0,1\}$ can be implemented in $\bigo{\frac{2^n}{n}}$ size and $\bigo{n}$ depth, with $4\cdot\frac{2^n}{n}+n\cdot 2^{\frac{n}{2}}$ ancilla.
\end{theorem}
\begin{proof}
    Let $q=\log n-1$ and $p=n-q$. To compute $f_S$ for all $S\in[p]$, it suffices to compute all $2^{2^q}$ Boolean functions of size $q$, $\mathcal{G}_q\coloneqq\{g:\{0,1\}^q\to\{0,1\}\}$:
    \begin{enumerate}
        \item compute all $Q\coloneqq\{\pi_S:S\subseteq[n]\setminus[p]\}$: $\bigo{2^q}$ size, $\bigo{q}$ depth with $2^q$ ancilla, according to \cref{lem:general_monomial};
        \item copy each element in $Q$ for $2^{2^q-1}$ times using fan-out gates in \cref{fact:fan}: size $\bigo{2^q2^{2^q}}$, and depth $\bigo{2^q}$, with $2^q2^{2^q-1}$ ancilla;
        \item compute all functions in $\mathcal{G}_q$: $\bigo{2^q2^{2^q}}$ size, and $\bigo{q}$ depth with $2^{2^q}$ ancilla, according to \cref{lem:general_from_monomial}.
    \end{enumerate}
    In total, we spend $\bigo{2^q2^{2^q}}$ size and $\bigo{2^q}$ depth, with $2^q+2^{2^q}+2^q2^{2^q}$ ancilla on computing all functions in $\mathcal{G}_q$.

    Now we are ready to implement $f$ using \cref{eq:general_partial}.
    \begin{enumerate}
        \item compute all $\{\pi_S:S\in[p]\}$: $\bigo{2^p}$ size, and $\bigo{p}$ depth, with $2^p$ ancilla, according to \cref{lem:general_monomial};
        \item copy each function $g$ in $\mathcal{G}_q$ for $\#\{S\in[p]:f_S\equiv g\}$ times using fan-out gates in \cref{fact:fan}: notice that elements in $\mathcal{G}_q$ will be copied for $2^p$ times in all, so this step has total size $\bigo{2^p}$, depth $\bigo{p}$, with $2^p$ ancilla;
        \item compute $f=\bigoplus_{S\subseteq[p]}\pi_S\wedge f_S$: $\bigo{2^p}$ size, $\bigo{p}$ depth with $0$ ancilla, by using Toffoli fan-in in \cref{fact:Toffoli_fanin}.
    \end{enumerate}
    In total, we spend $\bigo{2^p}$ size and $\bigo{p}$ depth, with $2^{p+1}$ ancilla, to compute $f$ from $\mathcal{G}_q$. To sum up, the whole construction has size $\bigo{2^p+2^q2^{2^q}}$, and depth $\bigo{p+2^q}$, with $2^{p+1}+2^q+2^{2^q}+2^q2^{2^q}$ ancilla. Plugging in the value of $p$ and $q$, we get the claimed result.
\end{proof}

\subsection{Complete depth-ancilla tradeoff}\label{ssec:general_tradeoff}
In this section, we prove \cref{thm:general_tradeoff}. We first settle the $m\ge n$ case in \cref{lem:general_tradeoff_large_space}, where we utilize \cref{eq:general_partial} again and the prefixes $f_S$ are computed using \cref{thm:general_optimal_size}. Finally for the $m<n$ case, we adopt the idea of conditionally clean ancilla from \cite{nie2024Quantum}.

\begin{lemma}\label{lem:general_tradeoff_large_space}
    Any Boolean function $f:\{0,1\}^n\to\{0,1\}$ can be implemented in $\bigo{\frac{2^n}{\log m}}$ size, $\bigo{\frac{2^n}{m}}$ depth, with $n\le m\le\bigtheta{\frac{2^n}{n}}$ ancilla.
\end{lemma}
\begin{proof}
    Let $q=\bigtheta{\log m+\log\log m}$ with $4\cdot \frac{2^q}{q}+q\cdot 2^{\frac{q}{2}}<m$, and $p=n-q$. We iterate each parity term $\pi_S\wedge f_S$ in \cref{eq:general_partial} to compute $f$. For each $S$,
    \begin{enumerate}
        \item compute $\pi_S$ using \cref{thm:Toffoli}: $\bigo{n}$ size, $\bigo{\log n}$ depth with $1$ ancilla;
        \item compute $f_S$ using \cref{thm:general_optimal_size}: $\bigo{m}$ size, $\bigo{\log m}$ depth within $m-1$ ancilla.
    \end{enumerate}
    Since $m\ge n$, each parity term can be implemented in size $\bigo{m}$, depth $\bigo{\log m}$ with $m$ ancilla. There are $2^p=\bigo{\frac{2^n}{m\log m}}$ terms, so this construction is as claimed.
\end{proof}

Now we are ready to present the proof of \cref{thm:general_tradeoff}.
\thmgeneraltradeoff*
\begin{proof}
    Compute the $b$-bit output one by one. Below, we assume that the output size is $1$.
    
    \cref{lem:general_tradeoff_large_space} already settles $m\ge n$. Now we assume $m<n$. Let $p+q=n$ and $p+m\le\bigtheta{\frac{2^q}{q}}$ as required by \cref{lem:general_tradeoff_large_space}, which can be achieved by letting $q=\bigtheta{\log(n+m)+\log\log(n+m)}$. Compute $f$ by the following formula:
    \[f=\oplus_{S\subseteq[p]}\sigma_S\wedge f_S(x_{p+1},\dots,x_n),\]
    where $\sigma_S=\bigwedge_{i\in S}x_i\bigwedge_{i\notin S}\neg x_i$. (Here we switch to $\sigma_S$ instead of $\pi_S$ mainly for the purpose of turning all the length $p$ prefix into conditionally clean ancilla.) For each $S$,
    \begin{enumerate}
        \item compute $\sigma_S$ into an ancillary qubit $A$ according to \cref{thm:Toffoli}: $\bigo{p}$ size, $\bigo{\log p}$ depth, with $1$ ancilla;
        \item condition on $A=1$, flip all the first $p$ input qubits to $\ket{0}$ using fan-out gate in \cref{fact:fan}: $\bigo{p}$ size, $\bigo{\log p}$ depth, without ancilla;
        \item compute the $q$-ary Boolean function $f_S$ into an ancillary qubit $B$ using \cref{lem:general_tradeoff_large_space}: $\bigo{\frac{2^q}{\log(p+m)}}$ size, $\bigo{\frac{2^q}{p+m}}$ depth, with $p+m\ge q$ ancilla where the first $p$ input qubits are treated as ancilla;
        \item compute $A\wedge B$ into the result qubit.
    \end{enumerate}
    For each $S$, the construction above has size $\bigo{p+\frac{2^q}{\log(p+m)}}$, depth $\bigo{\frac{2^q}{p+m}}$, with $m$ ancilla. For the correctness, notice that when $A=1$, all the first $p$ input qubits are returned to state $\ket{0}$ which is ready for being treated as clean ancilla; when $A=0$, $A\wedge B=0$ so the result qubit will not flip.

    The whole circuit iterates over $2^p$ possible choices of $S$, so the total size and depth is as claimed.
\end{proof}

The proof of \cref{lem:general_tradeoff_large_space} just invokes \cref{thm:general_optimal_size} for simplicity. In fact, all the $2^p$ sequential calls of \cref{thm:general_optimal_size} share the same Lupanov trick procedure that computes all possible Boolean functions, and the same monomial computation which invokes \cref{lem:general_monomial}. One can optimize the constant of the constructions in \cref{lem:general_tradeoff_large_space} and \cref{thm:general_tradeoff} further by expanding the proofs and carefully eliminating duplicate computations, which we omit here.

It is natural to construct the oracle of reversible Boolean functions with the help of \cref{thm:general_tradeoff}, which is \cref{cor:reversible_tradeoff}.
\correversibletradeoff*
\begin{proof}
    Since $f$ is reversible, $f^{-1}$ exists. Thus, $f$ can be computed in the following way:
    \[\ket{x,0^n,0^{m-n}}\mapsto\ket{x,f(x),0^{m-n}}\mapsto\ket{x\oplus f^{-1}(f(x)),f(x),0^{m-n}}.\]
    Each step calls \cref{thm:general_tradeoff} with $b=n$ once.
\end{proof}

\section{Partial Boolean oracle}\label{sec:partial}
In this section, we prove \cref{thm:partial_tradeoff}. For ancilla count $m$, the construction first uses linear hashing to reduce the input size from $n$ to $t=2\lceil\log d\rceil$ (\cref{lem:hash_distinct}), then uses another linear hashing with optimal balance to partition the size $d$ support to $\bigtheta{\frac{d}{m}}$ bins each of size $\bigtheta{m}$ (\cref{thm:hash_optimal_balance}). The depth and ancilla tradeoff of linear hashing is shown in \cref{lem:hash_tradeoff}. After that, it runs a partial oracle construction with rich ancilla which is stated in \cref{lem:partial_rich_anc}.

\subsection{Linear map as hash function}
When dealing with partial or sparse Boolean functions, we utilize linear map to hash the ``effective'' inputs so that the representation becomes shorter.
\begin{lemma}\label{lem:hash_distinct}
    Suppose $X=\{x_1,\dots,x_d\}\subseteq\{0,1\}^n$. Let $t=2\lceil\log d\rceil$. There exists a linear map $h:\mathbb{F}_2^n\to\mathbb{F}_2^t$ such that $h(x_i)$ are all distinct.
\end{lemma}
\begin{proof}
    Pick $h$ to be a random linear map. Let $\mathcal{H}$ be the set of all linear maps $\mathbb{F}_2^n\to\mathbb{F}_2^t$. For distinct $i,j$, we have
    \[\Pr_{h\sim\mathcal{H}}[h(x_i)=h(x_j)]=\Pr_{h\sim\mathcal{H}}[h(x_i\oplus x_j)=0]=2^{-t}.\]
    Thus, $\Pr[\exists i<j:h(x_i)=h(x_j)]\le\binom{d}{2}2^{-t}\le \frac{1}{2}$.
\end{proof}

Recently, Jaber, Kumar and Zuckerman \cite{jaber2025Linear} settles a long-open problem asking whether linear hashing is as good as random hashing with respect to the expectation of maximum load.
\begin{theorem}[\cite{jaber2025Linear}]\label{thm:hash_optimal_balance}
    Let $t\ge l$ and $d$ be integers, $L=2^l$, such that $d>\frac{lL}{2}$, and $\mathcal{H}$ be the set of linear maps $\mathbb{F}_2^t\to\mathbb{F}_2^l$. Then for any $Y\subseteq\mathbb{F}_2^t$ and $\#Y=d$,
    \[\mathbb{E}_{h\sim\mathcal{H}}\left[\max_{z\in\mathbb{F}_2^l}\# h^{-1}(z)\cap Y\right]\le 16\cdot d/L.\]
\end{theorem}

Such linear maps have simple efficient depth-ancilla tradeoff.
\begin{lemma}\label{lem:hash_tradeoff}
    For any linear map $h:\mathbb{F}_2^n\to\mathbb{F}_2^t$, the operator computing
    \[\ket{x}\ket{0}\to\ket{x}\ket{h(x)}\]
    can be implemented in size $\bigo{nt}$, depth $\bigo{\frac{nt\log n}{n+m}+\log(nt)}$, with $m$ ancilla.
\end{lemma}
\begin{proof}
    Suppose $m\le nt$. First, copy $x$ for $m/n$ times into the $m$ ancilla using fan-out gates in \cref{fact:fan}, which has size $\bigo{m}$, depth $\bigo{\log(m/n)}$. After that, compute $h(x)$ in batch with batch size $m/n+1$. For each $j\in[t]$, computing the $j$-th bit of $h(x)$ using fan-in gates has size $\bigo{n}$, depth $\bigo{\log n}$. So the total size is $\bigo{m+nt}=\bigo{nt}$, and the total depth is $\frac{t}{m/n+1}\cdot \bigo{\log n}=\bigo{\frac{nt\log n}{n+m}}$ as claimed.
\end{proof}

\subsection{The construction}
Every partial function $g$ of effective support size $s$ admits a decision tree of size $s$. When the ancilla count is large enough, we can compute $g$ by traversing its decision tree in parallel.
\begin{lemma}\label{lem:partial_rich_anc}
    Suppose the partial Boolean function $g:\{0,1\}^t\to\{0,1,*\}^b$ with $g^{-1}(\{0,1\}^b)=Y=\{y_1,\dots,y_s\}\subseteq\{0,1\}^t$. Then the oracle of $g$ can be implemented in size $\bigo{bs}$, depth $\bigo{t+b\log s}$, with $2s$ ancilla.
\end{lemma}
\begin{proof}
    It is clear that $g$ has a decision tree $T$ of size $s$ and depth at most $t$. For each node $u$ in $T$, associate two ancillary qubits, namely $Q_u$ and $C_u$, to $u$. First, query the input qubit that node $u$ asks onto $Q_u$. It can be implemented by $t$ fan-out gates in \cref{fact:fan}, so this step has size $s$, depth $\bigo{\log s}$, without ancilla.

    Next, we mark by $C_u$ the branch of $T$ the input qubits indicates. That is, $C_u=1$ if $u$ is on the branch of $T$ derived by input. To achieve this, we traverse by levels of $T$ from the root. For each node $u$ of $T$, suppose $v$ is $u$'s parent, then compute $C_v\wedge Q_u$ onto $C_u$. For each level the computation is merely some Toffoli gates with depth $2$, and the total number of levels is at most $t$. So This step has size $\bigo{s}$, depth $t$, without ancilla.

    Finally, compute $g$ controlled by all $C_u$ where $u$ is a leaf node using \cref{fact:fan}. This step has size $\bigo{bs}$, depth $\bigo{b\log s}$, without ancilla.
\end{proof}

Now we are ready to present the proof of \cref{thm:partial_tradeoff}.
\thmpartialtradeoff*
\begin{proof}
    Let $t=2\lceil\log d\rceil$, $L=2^l$ such that $L=\bigtheta{d/m}$ and $d\ge \frac{lL}{2}$ and $t+l+32\cdot d/L<m$. Because $d\ge \frac{lL}{2}$, it is clear that $t\ge l$ and $d/L=\bigtheta{m}$.
    
    Pick linear map $h_1:\mathbb{F}_2^n\to\mathbb{F}_2^t$ that satisfies \cref{lem:hash_distinct}. Denote $Y=\{y_1=h_1(x_1),\dots,y_d=h_1(x_d)\}$ and the induced partial function $g$ such that $g(y_i)=f(x_i)$. Pick another linear map $h_2:\mathbb{F}_2^t\to\mathbb{F}_2^l$ such that the maximum load satisfies \[\max_{z\in\mathbb{F}_2^l}\# h_2^{-1}(z)\cap Y\le 16\cdot d/L.\]

    The construction first computes $h_1$ and $h_2$, for each $x_i\in X$:
    \begin{align*}
        \ket{x_i}&\mapsto\ket{x_i}\ket{h_1(x_i)=y_i}\\
        &\mapsto\ket{x_i}\ket{y_i}\ket{h_2(y_i)}.
    \end{align*}
    According to \cref{lem:hash_tradeoff}, these two steps can be computed in $\bigo{nt+tl}=\bigo{nt}$ size, $\bigo{\frac{nt\log n}{n+32\cdot d/L}+\log n+\frac{tl\log t}{t+32\cdot d/L}+\log t}=\bigo{\frac{nt\log n}{n+m}+\log n}$ depth with $32\cdot d/L$ ancilla.

    For each $z\in\mathbb{F}_2^l$, denote the bin $Y_z=h_2^{-1}(z)\cap Y$:
    \begin{enumerate}
        \item compute a flag of $h_2(y_i)=z$ to an ancillary qubit $A$ using \cref{thm:Toffoli}, which has size $\bigo{l}$, depth $\bigo{\log l}$;
        \item Since $\# Y_z\le 16\cdot d/L$ and there are $32\cdot d/L$ ancillary qubits left, compute the restriction of $g$ on $Y_z$ using \cref{lem:partial_rich_anc}, which has size $\bigo{b\# Y_z}$ and depth $\bigo{t+b\log\# Y_z}$;
        \item Take the result of $g$ only if $A=1$.
    \end{enumerate}
    There are $L$ bins, so the total size of this step is $\bigo{lL+bd}=\bigo{bd}$, and the total depth is $\bigo{L(\log l+t+b\log(32\cdot d/L))}=\bigo{L(t+b\log m)}$.

    In all, the construction has size $\bigo{nt+bd}=\bigo{n\log d+bd}$, and depth \[\bigo{\frac{nt\log n}{n+m}+\log n+L(t+b\log m)}=\bigo{\frac{n\log n\log d}{n+m}+\log n+\frac{d(\log d+b\log m)}{m}}\] as claimed.    
\end{proof}

For certain output size $b$ and enough ancilla count $m$, one can further reduce the depth dependence of $b$ to as least as $\log b$ by providing sufficient ancilla to the last step in \cref{lem:partial_rich_anc}, which we omit here.

\section{Sparse Boolean oracle}\label{sec:sparse}
Given a sparse total Boolean function $f:\{0,1\}^n\to\{0,1\}$ with $f^{-1}(1)=X=\{x_1,\dots,x_d\}$. In this section, we present our construction for sparse Boolean oracle of $f$.

\subsection{Linear map as subset separator}
Hash functions may be used to separate $X$ from $\{0,1\}^n\setminus X$, which we call it $X$-separating. For example, a random hash function in $\mathbb{F}_2^n\to\mathbb{F}_2^r$ separates $x\notin X$ from $X$ with probability $d\cdot 2^{-r}$. So there exists roughly $n$ hash functions in $\mathbb{F}_2^n\to\mathbb{F}_2^{\log d}$, such that the membership checking of their images is equivalent to the membership checking of $X$. Set separating hash family is a well-studied object, which requires the family separating all set tuples that satisfies certain property. Compared to set separating, $X$-separating only focuses on separating one particular set $X$ and its complement $\{0,1\}^n\setminus X$, so the size of hash images can be very small:

\begin{lemma}\label{lem:hash_separator}
    Suppose $X=\{x_1,\dots,x_d\}\subseteq\mathbb{F}_2^n$. Let $\kappa=n+1$ and $r=\log d+1$. There exists linear maps $h_1,\dots,h_\kappa:\mathbb{F}_2^n\to\mathbb{F}_2^r$ such that \[X=\{x\in\{0,1\}^n:\forall j\in[\kappa],h_j(x)\in h_j(X)\}.\]
\end{lemma}
\begin{proof}
    Random pick $h_1,\dots,h_\kappa$ from linear maps $\mathbb{F}_2^n\to\mathbb{F}_2^r$. For any $x\notin X$,
    \[\Pr[\forall j\in[\kappa],\exists i\in[d],h_j(x)=h_j(x_i)]=\prod_{j\in[\kappa]}\Pr[\exists i\in[d],h_j(x)=h_j(x_i)]\le (d\cdot 2^{-r})^\kappa.\]
    Thus, the failure probability of random $h_1,\dots,h_\kappa$ is at most $2^n\cdot(d\cdot 2^{-r})^\kappa=\frac{1}{2}$.
\end{proof}

Random linear maps are already good at $X$-separating. If we first encode the input by a linear code $C$ with constant relative distance, say $[\bigo{n},n,\delta]$, then the linear maps can be further reduced to ones with low row weight. By low row weight, we mean that the Hamming weight of each row is small in the matrix representation of the hash map.
\begin{lemma}\label{lem:hash_separator_improved}
    Suppose $X=\{x_1,\dots,x_d\}\subseteq\mathbb{F}_2^n$. Let $C:\mathbb{F}_2^n\to\mathbb{F}_2^l$ be a linear code with $l=\bigo{n}$ and relative distance $\delta$. Let $\kappa=n+1$ and $r=\log d+\bigo{1}$. There exists linear maps $h_1,\dots,h_\kappa:\mathbb{F}_2^n\to\mathbb{F}_2^r$ where $h_j=H_jC$ for some $H_j$ of row weight $\bigo{\log\log d}$, and \[X=\{x\in\{0,1\}^n:\forall j\in[\kappa],h_j(x)\in h_j(X)\}.\]
\end{lemma}
\begin{proof}
    Let $w$ be an odd number such that $w=\bigtheta{\log\log d}$ and $(1-2\delta)^w\le\log^{-1} d$. Pick $H_1,\dots,H_\kappa$ uniform randomly from linear maps $\mathbb{F}_2^l\to\mathbb{F}_2^r$ that has row weight $w$. Then for any $x\notin X$,
    \[\Pr[h_j(x)=h_j(x_i)]=\Pr[H_jCx=H_jCx_i]\le\left[\frac{1+(1-2\delta)^w}{2}\right]^r\le\left(\frac{1+\log^{-1} d}{2}\right)^r.\]
    Thus, there exists $r=\log d+\bigo{1}$ such that
    \[\Pr[\exists i\in[d],h_j(x)=h_j(x_i)]\le d\cdot\left(\frac{1+\log^{-1} d}{2}\right)^r\le \frac{1}{2}.\]
    So for one $x\notin X$, $h_1,\dots,h_\kappa$ fail to separate $x$ from $X$ with probability no more than $2^{-\kappa}$, which means the total failure probability is at most $\frac{1}{2}$.
\end{proof}

When the ancilla count is large enough, we can compute all hashes and check memberships in parallel.
\begin{lemma}\label{lem:sparse_membership_rich_anc}
    Given linear maps $h_1,\dots,h_\mu:\mathbb{F}_2^n\to\mathbb{F}_2^r$, and $Y_1,\dots,Y_\mu\subseteq\mathbb{F}_2^r$. The membership check 
    \[\ket{x}\ket{b}\to\ket{x}\ket{b\oplus\bigwedge_{j=1}^\mu[h_j(x)\in Y_j]}\]
    can be implemented in size $\bigo{\mu nr+\frac{\mu2^r}{\log(m/\mu)}}$, and depth $\bigo{\frac{\mu nr\log n}{n+m}+\frac{\mu2^r}{m}+\log(\mu nr)}$, with $\mu(r+1)<m\le\mu\cdot\bigtheta{\frac{2^r}{r}}$ ancilla.
\end{lemma}
\begin{proof}
    Let $a=m/\mu-r$ be the ancilla count rest for manipulation of each $h_j$, so $a\le\bigtheta{\frac{2^r}{r}}$. The construction has the following three steps:
    \begin{enumerate}
        \item compute $h_1,\dots,h_\mu$ into $\mu r$ ancillary qubits using \cref{lem:hash_tradeoff}, which has size $\bigo{\mu nr}$, and depth $\bigo{\frac{\mu nr\log n}{n+m}+\log (\mu nr)}$, with $m$ ancilla;
        \item check whether $h_j(x)\in Y_j$ for all $j$ using $\mu$ parallel invocations of total Boolean oracle construction in \cref{thm:general_tradeoff}, which has size $\bigo{\frac{\mu2^r}{\log(r+a)}}$, depth $\bigo{\frac{2^r}{r+a}}$, with $\mu a$ ancilla;
        \item use $\mu$-Toffoli gate to compute conjunction, which has size $\bigo{\mu}$, depth $\bigo{\log\mu}$ according to \cref{thm:Toffoli}.
    \end{enumerate}
    In all, the construction has size $\bigo{\mu nr+\frac{\mu2^r}{\log(m/\mu)}}$, and depth $\bigo{\frac{\mu nr\log n}{n+m}+\frac{\mu2^r}{m}+\log(\mu nr)}$, with $m$ ancilla as claimed.
\end{proof}

When the ancilla count is not large enough to compute all hashes in parallel, we divide them into batches of, say, size $\mu$. However, to recycle the ancillary qubits after each batch, the naive way is to memorize each of the $\frac{\kappa}{\mu}$ batch membership check result, and it would cost $\frac{\kappa}{\mu}$ extra ancilla. In the extreme sequential case, the minimum ancilla count requirement would still be $\bigomega{n}$, not logarithmic. To address this small issue, a special version of low-width Toffoli gates whose inputs are queried from sequential oracles is needed. The existing low-width constructions of Toffoli gate are not helpful because in these constructions input qubits not only serve as control qubits.
\begin{lemma}\label{lem:sparse_Toffoli_oracle}
    Given oracles $O_1,\dots,O_\nu$ acting on the same quantum register $\ket{z,b}$ such that
    \[O_j:\ket{z,b}\to\ket{z,b\oplus O_j(z)}.\]
    There is a quantum circuit computing $\wedge_{j\in[\nu]}O_j(z)$ with $2\nu$ sequential calls of the oracles, in size and depth $\bigo{\nu\log\nu}$, with $\log\nu+1$ ancilla.
\end{lemma}
\begin{proof}
    Denote one ancillary qubit $A$, and the rest $\log\nu$ ancilla register $B$. Query $O_1,\dots,O_\nu$ one by one in sequential. The idea is to maintain $B$ to be the first $j$ queried such that $O_j(z)=0$, so that $A$ will be flipped at most once. After the $j$-th query:
    \begin{enumerate}
        \item if $O_j(z)=0$ and $A=0$, write the binary of $j$ into $B$ using $\log\nu$ Toffoli gates;
        \item flip $A$ conditioned on $O_j(z)=0$ and $B=j$.
    \end{enumerate}
    It is clear that $\neg A=\wedge_j O_j(z)$.
\end{proof}

Now we are ready to present the proof of \cref{thm:sparse_tradeoff}.
\thmsparsetradeoff*
\begin{proof}
    Let $\kappa=n+1$ and $r=\log d+1$, let $\mu=\min\{\bigtheta{\frac{m}{r}},\kappa\}$ and $\nu\mu=\kappa$, which implies $\nu=\bigtheta{1+\frac{nr}{m}}$.
    
    Pick linear maps $h_1,\dots,h_\kappa\in\mathbb{F}_2^n\to\mathbb{F}_2^r$ satisfying \cref{lem:hash_separator}, and divide them into $\nu=\frac{\kappa}{\mu}$ batches each of size $\mu$. Construct oracles $O_1,\dots, O_\nu$ according to \cref{lem:sparse_membership_rich_anc}, each of which has size $\bigo{\mu nr+\frac{\mu2^r}{\log(m/\mu)}}$, and depth $\bigo{\frac{\mu nr\log n}{n+m}+\frac{\mu2^r}{m}+\log(\mu nr)}$. Use the special Toffoli gate in \cref{lem:sparse_Toffoli_oracle} to compute $\wedge_{j\in[\nu]}O_j$. According to \cref{lem:hash_separator}, it checks memberships of $X$-separating sets, so it equals $f$.
    
    The overall construction has size \[2\nu\cdot\bigo{\mu nr+\frac{\mu2^r}{\log(m/\mu)}}+\bigo{\nu\log\nu}\] and depth \[2\nu\cdot\bigo{\frac{\mu nr\log n}{n+m}+\frac{\mu2^r}{m}+\log(\mu nr)}+\bigo{\nu\log\nu}.\]

    When $\frac{m}{r}\le\kappa$, we have $m<n\log d$ and $\mu=\bigtheta{\frac{m}{r}}$, thus $ m/\mu=r<\frac{2^r}{r}$, which means the $\mu$ calls of \cref{lem:sparse_membership_rich_anc} are proper. In this case, the size is $\bigo{n^2\log d+\frac{nd}{\log\log d}}$, and the depth is $\bigo{\frac{n^2\log n\log d}{n+m}+\frac{nd}{m}}$.

    When $\frac{m}{r}>\kappa$, we have $m\ge n\log d$ and $\mu=n$, thus $ m/\mu\le\frac{2^r}{r}$ since $m\le\bigtheta{\frac{nd}{\log d}}$ which means the $\mu$ calls of \cref{lem:sparse_membership_rich_anc} are also proper. In this case, the size is $\bigo{n^2\log d+\frac{nd}{\log(m/n)}}$, and the depth is $\bigo{\frac{n^2\log n\log d}{n+m}+\log n+\frac{nd}{m}}$.
\end{proof}

\subsection{Further improvement}
The $\bigo{n^2\log d}$ term of the size comes from computing $n$ random hash functions in $\mathbb{F}_2^n\to\mathbb{F}_2^{\log d}$. If $\bigo{n}$ extra ancilla is available, one can first encode the input by a linear code $C$ with constant relative distance. \cref{lem:hash_separator_improved} guarantees that after encoding the input by $C$, one can replace the naive random hash functions by low-density ones. This results in a more efficient construction for the hashing step, so it only improves the asymptotic size and depth when the sparsity $d$ is small. The proof is similar to \cref{thm:sparse_tradeoff}, so we omit the details.

Low weight linear functions can be implemented by a much more efficient construction.
\begin{lemma}\label{lem:hash_tradeoff_low_weight}
    For any linear map $H:\mathbb{F}_2^l\to\mathbb{F}_2^t$ with row weight $w\le t$, the operator computing
    \[\ket{x}\ket{0}\to\ket{x}\ket{H(x)}\]
    can be implemented in size $\bigo{wt}$, depth $\bigo{\frac{wt\log w}{w+m}+\log t}$, with $m$ ancilla.
\end{lemma}
\begin{proof}[Proof Sketch]
    The proof is similar to \cref{lem:hash_tradeoff}.
\end{proof}

Again, when the ancilla count is large enough, we can compute all low-weight hash functions in parallel.
\begin{lemma}\label{lem:sparse_membership_rich_anc_improved}
    Given linear maps $H_1,\dots,H_\mu:\mathbb{F}_2^l\to\mathbb{F}_2^r$ of row weight $w$, and $Y_1,\dots,Y_\mu\subseteq\mathbb{F}_2^r$. The membership check 
    \[\ket{x}\ket{b}\to\ket{x}\ket{b\oplus\bigwedge_{j=1}^\mu[H_j(x)\in Y_j]}\]
    can be implemented in size $\bigo{\mu wr+\frac{\mu2^r}{\log(m/\mu)}}$, and depth $\bigo{\frac{\mu wr\log w}{w+m}+\frac{\mu2^r}{m}+\log(\mu wr)}$, with $\mu(r+1)<m\le\mu\cdot\bigtheta{\frac{2^r}{r}}$ ancilla.
\end{lemma}
\begin{proof}[Proof Sketch]
    The proof is similar to \cref{lem:sparse_membership_rich_anc}. The only difference is to replace the first step of \cref{lem:sparse_membership_rich_anc} with \cref{lem:hash_tradeoff_low_weight}.
\end{proof}

\begin{theorem}\label{thm:sparse_tradeoff_improved}
    Suppose the sparse total Boolean function $f:\{0,1\}^n\to\{0,1\}$ with $f^{-1}(1)=X=\{x_1,\dots,x_d\}$. Then the oracle of $f$ can be implemented in size $\tildeo{n^2+\frac{nd}{\log(\log d+m/n)}}$, depth $\bigo{\frac{n^2\log n}{n+m}+\log n+\frac{nd}{m}}$, with $\bigtheta{n}+m$ ancilla where $\bigtheta{\log(nd)}\le m\le\bigtheta{\frac{nd}{\log d}}$.
\end{theorem}
\begin{proof}[Proof Sketch]
    The construction is similar to \cref{thm:sparse_tradeoff}. We adopt the parameters from \cref{thm:sparse_tradeoff}. Pick $h_1,\dots,h_\kappa$ satisfying \cref{lem:hash_separator_improved}. Denote $h_j=H_jC$, where $C:\mathbb{F}_2^n\to\mathbb{F}_2^l,l=\bigo{n}$ is a linear code and $H_j:\mathbb{F}_2^l\to\mathbb{F}_2^r$ are of row weight $w=\bigtheta{\log\log d}$.
    
    According to \cref{lem:hash_tradeoff}, the linear code $C$ can be implemented in size $\bigo{nl}$ and depth $\bigo{\frac{nl\log n}{n+m}+\log(nl)}$ with $m$ ancilla. Using \cref{lem:sparse_membership_rich_anc_improved}, each of the oracles $O_1,\dots,O_\nu$ can be implemented in size $\bigo{\mu wr+\frac{\mu2^r}{\log(m/\mu)}}$, and depth $\bigo{\frac{\mu wr\log w}{w+m}+\frac{\mu2^r}{m}+\log(\mu wr)}$.
    
    Thus, the overall construction has size \[\bigo{nl}+2\nu\cdot\bigo{\mu wr+\frac{\mu2^r}{\log(m/\mu)}}+\bigo{\nu\log\nu}\] and depth \[\bigo{\frac{nl\log n}{n+m}+\log(nl)}+2\nu\cdot\bigo{\frac{\mu wr\log w}{w+m}+\frac{\mu2^r}{m}+\log(\mu wr)}+\bigo{\nu\log\nu}.\]

    When $\frac{m}{r}\le\kappa$, the overall construction has size $\bigo{n^2+n\log d\log\log d+\frac{nd}{\log\log d}}$, and depth $\bigo{\frac{n^2\log n}{n+m}+\frac{nd}{m}}$.

    When $\frac{m}{r}>\kappa$, it has size $\bigo{n^2+n\log d\log\log d+\frac{nd}{\log(m/n)}}$, and depth $\bigo{\frac{n^2\log n}{n+m}+\log n+\frac{nd}{m}}$.

\end{proof}

\bibliographystyle{alpha}
\bibliography{citations.bib}

\end{document}